\documentclass[10pt,conference]{IEEEtran}
\usepackage{epsfig}
\usepackage{graphicx}
\usepackage{cite}
\usepackage{amsmath}
\usepackage{amssymb}

\newcommand{\comment}[1]  {}
\def\BE{\begin{equation}}
\def\EE{\end{equation}}
\def\BEA{\begin{eqnarray}}
\def\EEA{\end{eqnarray}}

\DeclareMathOperator*{\argmax}{arg\,max}
\newtheorem{thm}{Theorem}
\newtheorem{lem}[thm]{Lemma}
\newtheorem{alg}{Algorithm}
\newtheorem{prop}[thm]{Proposition}

\newtheorem{corol}[thm]{Corollary}
\newtheorem{remark}[thm]{Remark}

\newcommand\ie{{\textsl{i.e.\,}}}
\newcommand\eg{{\textsl{e.g.\,}}}
\newcommand\etal{{\textsl{et al.\,}}}

\newcommand\vb{{\bf b}}

\newcommand\vx{{\bf x}}
\newcommand\vy{{\bf y}}

\newcommand\mA{{\bf A}} %matrix

\newcommand\mI{{\bf I}}

\newcommand\mR{{\bf R}}

\begin{document}

\title{Gaussian Belief Propagation Solver\\for Systems of Linear Equations}

\author{\authorblockN{Ori Shental$^{1}$, Paul H. Siegel and Jack K. Wolf}
\authorblockA{Center for Magnetic Recording Research\\
University of California - San Diego\\
La Jolla, CA 92093, USA\\
Email: \{oshental,psiegel,jwolf\}@ucsd.edu}
\and
\authorblockN{Danny Bickson$^{1}$ and Danny Dolev}
\authorblockA{School of Computer Science and Engineering\\
Hebrew University of Jerusalem\\
Jerusalem 91904, Israel\\
Email: \{daniel51,dolev\}@cs.huji.ac.il}}

\maketitle

\begin{abstract}
\footnotetext[1]{Contributed equally to this work.\\Supported in part by NSF Grant
No.~CCR-0514859 and EVERGROW, IP 1935 of the EU Sixth Framework.}
The canonical problem of solving a system of linear equations
arises in numerous contexts in information theory, communication
theory, and related fields. In this contribution, we develop a
solution based upon Gaussian belief propagation (GaBP) that does
not involve direct matrix inversion. The iterative nature of our
approach allows for a distributed message-passing implementation
of the solution algorithm. We also address some properties of the
GaBP solver, including convergence, exactness, its max-product
version and relation to classical solution methods. The
application example of decorrelation in CDMA is used to
demonstrate the faster convergence rate of the proposed solver in
comparison to conventional linear-algebraic iterative solution
methods.
\end{abstract}

\section{Problem Formulation and Introduction}
%The setting
Solving a system of linear equations \mbox{$\mA\vx=\vb$} is one of
the most fundamental problems in algebra, with countless
applications in the mathematical sciences and engineering. Given
the observation vector
\mbox{$\vb\in\mathbb{R}^{n},n\in\mathbb{N}^{\ast}$}, and the data
matrix \mbox{$\mA\in\mathbb{R}^{n\times n}$}, a unique solution,
\mbox{$\vx=\vx^{\ast}\in\mathbb{R}^{n}$}, exists if and only if
the data matrix $\mA$ is full rank. In this contribution we
concentrate on the popular case where the data matrices, $\mA$,
are also symmetric (\eg, as in correlation matrices). \comment{A
possible extension of our approach to the non-symmetric case, or
even rectangular data matrices, is concisely addressed in
Section~\ref{}.}
Thus, assuming a nonsingular symmetric matrix $\mA$, the system of
equations can be solved either directly or in an iterative manner.
Direct matrix inversion methods, such as Gaussian elimination (LU
factorization,~\cite{BibDB:BookMatrix}-Ch. 3) or band Cholesky
factorization (\cite{BibDB:BookMatrix}-Ch. 4), find the solution
with a finite number of operations, typically, for a dense
$n\times n$ matrix, on the order of $n^{3}$. The former is
particularly effective for systems with unstructured dense data
matrices, while the latter is typically used for structured dense
systems.

Iterative methods~\cite{BibDB:BookAxelsson} are inherently
simpler, requiring only additions and multiplications, and have
the further advantage that they can exploit the sparsity of the
matrix $\mA$ to reduce the computational complexity as well as the
algorithmic storage requirements~\cite{BibDB:BookSaad}. By
comparison, for large, sparse and amorphous data matrices, the
direct methods are impractical due to the need for excessive row
reordering operations.
The main drawback of the iterative approaches is that, under
certain conditions, they converge only asymptotically to the exact
solution $\vx^{\ast}$~\cite{BibDB:BookAxelsson}. Thus, there is
the risk that they may converge slowly, or not at all. In
practice, however, it has been found that they often converge to
the exact solution or a good approximation after a relatively
small number of iterations.

A powerful and efficient iterative algorithm, belief propagation
(BP)~\cite{BibDB:BookPearl}, also known as the sum-product
algorithm, has been very successfully used to solve, either
exactly or approximately, inference problems in probabilistic
graphical models~\cite{BibDB:BookJordan}.
%Statement of purpose and value
In this paper, we reformulate the general problem of solving a
linear system of algebraic equations as a probabilistic inference
problem on a suitably-defined graph. We believe that this is the
first time that an explicit connection between these two
ubiquitous problems has been established. As an important
consequence, we demonstrate that Gaussian BP (GaBP) provides an
efficient, distributed approach to solving a linear system that
circumvents the potentially complex operation of direct matrix
inversion. \comment{Using the seminal work of Weiss and
Freeman~\cite{BibDB:Weiss01Correctness} and related recent
developments~\cite{BibDB:jmw_walksum_nips,BibDB:mjw_walksum_jmlr06},
we address the convergence and exactness properties of the
proposed GaBP solver. The GaBP solver's max-product version and
the relation to classical solution methods are also investigated.}

We shall use the following notations. The operator $\{\cdot\}^{T}$
denotes a vector or matrix transpose, the matrix $\mI_{n}$ is a
$n\times n$ identity matrix, while the symbols $\{\cdot\}_{i}$ and
$\{\cdot\}_{ij}$ denote entries of a vector and matrix,
respectively.

%The work of Weiss~\cite{Linear1,Linear2} is the only work we are
%aware of that discusses the connection between belief propagation
%and linear programming. Complementary to our construction where we
%deal with continues probabilities,Weiss deals with discrete
%variables. The problem, as in our case, is to find x* that
%maximizes the posterior probability.  This problem is NP-hard. The
%following construction is used: first, the problem is shifted to
%an equivalent integer programming formulation, by adding indicator
%variables for the variables and the potentials. Second, a LP
%relaxation is done, by allowing the indicator variables to be
%bounded in [0,1], thus taking non-integer values. This makes the
%problem solvable in polynomial   time. In the special case where
%the solution computed has only integers, the solution as an
%optimal solution to the MAP assignment.
%
%There are two possible limitations to this approach: first, the
%number of indicator variables grows exponentially with the number
%of states. Second, it solves the MAP assignment only in special
%cases, for example when the approximated free energy function is
%convex.
%
%%\bibitem{Linear2}
%
%Yair Weiss, Chen Yanover, Talya Meltzer MAP Estimation, Linear
%Programming and Belief Propagation with Convex Free Energies. In
%UAI 2007.
%%\bibitem{Linear1}
%Chen Yanover, Talya Meltzer, Yair Weiss. Linear Programming
%Relaxations and Belief Propagation - an Empirical Study JMLR
%Special Issue on Machine Learning and Large Scale Optimization,
%2005

\section{The GaBP Solver}\label{sec_GaBP}

\subsection{From Linear Algebra to Probabilistic Inference}
We begin our derivation by defining an undirected graphical model
(\ie, a Markov random field), $\mathcal{G}$, corresponding to the
linear system of equations. Specifically, let
$\mathcal{G}=(\mathcal{X},\mathcal{E})$, where $\mathcal{X}$ is a
set of nodes that are in one-to-one correspondence with the linear
system's variables $\vx=\{x_{1},\ldots,x_{n}\}^{T}$, and where
$\mathcal{E}$ is a set of undirected edges determined by the
non-zero entries of the (symmetric) matrix $\mA$. Using this
graph, we can translate the problem of solving the linear system
from the algebraic domain to the domain of probabilistic
inference, as stated in the following theorem.

\begin{prop}[Solution and inference]\label{prop_3}
The computation of the solution vector $\vx^{\ast}$ is identical
to the inference of the vector of marginal means
\mbox{$\mathbf{\mu}=\{\mu_{1},\ldots,\mu_{n}\}$} over the graph
$\mathcal{G}$ with the associated joint Gaussian probability
density function
\mbox{$p(\vx)\sim\mathcal{N}(\mu\triangleq\mA^{-1}\vb,\mA^{-1})$}.
\end{prop}
\begin{proof}
Another way of solving the set of linear equations
$\mA\vx-\vb=\mathbf{0}$ is to represent it by using a quadratic
form \mbox{$q(\vx)\triangleq\vx^{T}\mA\vx/2-\vb^{T}\vx$}. As the
matrix $\mA$ is symmetric, the derivative of the quadratic form
w.r.t. the vector $\vx$ is given by the vector $\partial
q/\partial\vx=\mA\vx-\vb$. Thus equating $\partial q/\partial
\vx=\mathbf{0}$ gives the stationary point $\vx^{\ast}$, which is
nothing but the desired solution to $\mA\vx=\vb$.
Next, one can define the following joint Gaussian probability
density function \BE\label{eq_G}
p(\vx)\triangleq\mathcal{Z}^{-1}\exp{\big(-q(\vx)\big)}=\mathcal{Z}^{-1}\exp{(-\vx^{T}\mA\vx/2+\vb^{T}\vx)},\EE
where $\mathcal{Z}$ is a distribution normalization factor.
Denoting the vector $\mathbf{\mu}\triangleq\mA^{-1}\vb$, the
Gaussian density function can be rewritten as \BEA\label{eq_G2}
p(\vx)&=&\mathcal{Z}^{-1}\exp{(\mathbf{\mu}^{T}\mA\mathbf{\mu}/2)}\nonumber\\&\times&\exp{(-\vx^{T}\mA\vx/2+\mathbf{\mu}^{T}\mA\vx-\mathbf{\mu}^{T}\mA\mathbf{\mu}/2)}
\nonumber\\&=&\mathcal{\zeta}^{-1}\exp{\big(-(\vx-\mathbf{\mu})^{T}\mA(\vx-\mathbf{\mu})/2\big)}\nonumber\\&=&\mathcal{N}(\mathbf{\mu},\mA^{-1}),\EEA
where the new normalization factor
$\mathcal{\zeta}\triangleq\mathcal{Z}\exp{(-\mathbf{\mu}^{T}\mA\mathbf{\mu}/2)}$.
It follows that the target solution $\vx^{\ast}=\mA^{-1}\vb$ is
equal to $\mathbf{\mu}\triangleq\mA^{-1}\vb$, the mean vector of
the distribution $p(\vx)$, as defined above~(\ref{eq_G}).
Hence, in order to solve the system of linear equations we need to
infer the marginal densities, which must also be Gaussian,
\mbox{$p(x_{i})\sim\mathcal{N}(\mu_{i}=\{\mA^{-1}\vb\}_{i},P_{i}^{-1}=\{\mA^{-1}\}_{ii})$},
where $\mu_{i}$ and $P_{i}$ are the marginal mean and inverse
variance (sometimes called the precision), respectively.
\end{proof}

According to Proposition~\ref{prop_3}, solving a deterministic
vector-matrix linear equation translates to solving an inference
problem in the corresponding graph. The move to the probabilistic
domain calls for the utilization of BP as an efficient inference
engine.

\comment{
\begin{remark}
Defining a jointly Gaussian probability density function
immediately yields an implicit assumption on the positive
semi-definiteness of the precision matrix $\mA$, in addition to
the symmetry assumption. However, we would like to stress that
this assumption emerges only for exposition purposes, so we can
use the notion of `Gaussian probability', but the derivation of
the GaBP solver itself does not use this assumption.\comment{See
the numerical example of the exact GaBP-based solution of a system
with a symmetric, but not positive semi-definite, data matrix
$\mA$ in Section~\ref{sec_nonPSD}.}
\end{remark}}
\subsection{Belief Propagation in Graphical Model}
Belief propagation (BP) is equivalent to applying Pearl's local
message-passing algorithm~\cite{BibDB:BookPearl}, originally
derived for exact inference in trees, to a general graph even if
it contains cycles (loops). BP has been found to have outstanding
empirical success in many applications, \eg, in decoding Turbo
codes and low-density parity-check (LDPC) codes. The excellent
performance of BP in these applications may be attributed to the
sparsity of the graphs, which ensures that cycles in the graph are
long, and inference may be performed as if the graph were a tree.

Given the data matrix $\mA$ and the observation vector $\vb$, one
can write explicitly the Gaussian density function,
$p(\vx)$~(\ref{eq_G2}), and its corresponding graph $\mathcal{G}$
consisting of edge potentials (`compatibility functions')
$\psi_{ij}$ and self potentials (`evidence') $\phi_{i}$. These
graph potentials are simply determined according to the following
pairwise factorization of the Gaussian function~(\ref{eq_G})\BE
p(\vx)\propto\prod_{i=1}^{n}\phi_{i}(x_{i})\prod_{\{i,j\}}\psi_{ij}(x_{i},x_{j}),\EE
resulting in \mbox{$\psi_{ij}(x_{i},x_{j})\triangleq
\exp(-x_{i}A_{ij}x_{j})$} and
\mbox{$\phi_{i}(x_{i})\triangleq\exp\big(b_{i}x_{i}-A_{ii}x_{i}^{2}/2\big)$}.
Note that by completing the square, one can observe that \mbox{$
\phi_{i}(x_{i})\propto\mathcal{N}(\mu_{ii}=b_{i}/A_{ii},P_{ii}^{-1}=A_{ii}^{-1})$}.
The graph topology is specified by the structure of the matrix
$\mA$, \ie, the edges set $\{i,j\}$ includes all non-zero entries
of $\mA$ for which $i>j$.

The BP algorithm functions by passing real-valued messages across
edges in the graph and consists of two computational rules, namely
the `sum-product rule' and the `product rule'. In contrast to
typical applications of BP in coding theory~\cite{BibDB:BookMCT},
our graphical representation resembles a pairwise Markov random
field\cite{BibDB:BookJordan} with a single type of propagating
message, rather than a factor graph~\cite{BibDB:FactorGraph} with
two different types of messages, originating from either the
variable node or the factor node. Furthermore, in most graphical
model representations used in the information theory literature
the graph nodes are assigned discrete values, while in this
contribution we deal with nodes corresponding to continuous
variables. Thus, for a graph $\mathcal{G}$ composed of potentials
$\psi_{ij}$ and $\phi_{i}$ as previously defined, the conventional
sum-product rule becomes an integral-product
rule~\cite{BibDB:Weiss01Correctness} and the message
$m_{ij}(x_j)$, sent from node $i$ to node $j$ over their shared
edge on the graph, is given by \BE\label{eq_contBP}
    m_{ij}(x_j)\propto\int_{x_i} \psi_{ij}(x_i,x_j) \phi_{i}(x_i)
\prod_{k \in \textrm{N}(i)\setminus j} m_{ki}(x_i) dx_{i}. \EE The
marginals are computed (as usual) according to the product rule
\BE\label{eq_productrule} p(x_{i})=\alpha
\phi_{i}(x_{i})\prod_{k\in\textrm{N}(i)}m_{ki}(x_{i}), \EE where
the scalar $\alpha$  is a normalization constant. The set of graph
nodes $\textrm{N}(i)$ denotes the set of all the nodes neighboring
the $i$th node. The set $\textrm{N}(i)\backslash j$ excludes the
node $j$ from $\textrm{N}(i)$.

\subsection{The Gaussian BP Algorithm}\label{sec:GaBP}
Gaussian BP is a special case of continuous BP, where the
underlying distribution is Gaussian. Now, we derive the Gaussian
BP update rules by substituting Gaussian distributions into the
continuous BP update
equations~(\ref{eq_contBP})-(\ref{eq_productrule}).
Before describing the inference algorithm performed over the
graphical model, we make the elementary but very useful
observation that the product of Gaussian densities over a common
variable is, up to a constant factor, also a Gaussian density.

\begin{lem}[Product of Gaussians]\label{Lemma_Gaussian}
Let $f_{1} (x)$ and $f_{2} (x)$ be the probability density
functions of a Gaussian random variable with two possible
densities $\mathcal{N}(\mu_{1},P_{1}^{-1})$ and
$\mathcal{N}(\mu_{2},P_{2}^{-1})$, respectively. Then their
product, \mbox{$f(x)=f_1(x) f_2(x)$} is, up to a constant factor,
the probability density function of a Gaussian random variable
with distribution $\mathcal{N}(\mu,P^{-1})$, where \BEA
P^{-1}&=&(P_{1}+P_{2})^{-1}\label{eq_productprec},\\
\mu&=&P^{-1}(P_{1}\mu_{1}+P_{2}\mu_{2})\label{eq_productmean}.
\EEA
\end{lem}
\begin{proof}
The proof of this lemma is straightforward, thus omitted.
\end{proof}

\begin{figure}[thb!]
\begin{center}
    \includegraphics[width=0.2\textwidth]{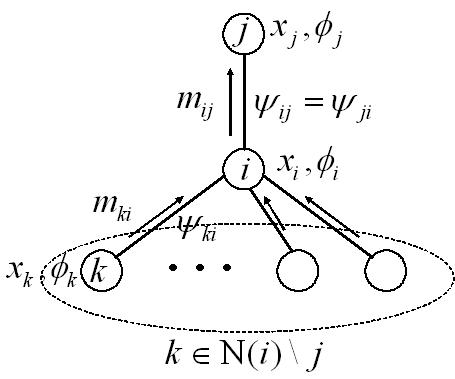}
   \vspace{-0.4cm}\caption{Graphical model: The neighborhood of node $i$.}
\end{center}\label{fig_flow}
\end{figure}
Fig. 1.\comment{~\ref{fig_flow}} plots a portion of a certain
graph, describing the neighborhood of node $i$. Each node (empty
circle) is associated with a variable and self potential $\phi$,
which is a function of this variable, while edges are identified
with the pairwise (symmetric) potentials $\psi$. Messages
propagate along the edges in both directions.
The messages relevant for the computation of message $m_{ij}$ are shown in Fig.~1.\comment{~\ref{fig_flow}}).
Looking at the right hand side of the integral-product
rule~(\ref{eq_contBP}), node $i$ needs to first calculate the
product of all incoming messages, except for the message coming
from node $j$. Recall that since $p(\vx)$ is jointly Gaussian, the
factorized self potentials
$\phi_{i}(x_i)\propto\mathcal{N}(\mu_{ii},P_{ii}^{-1})$ and
similarly all messages
$m_{ki}(x_i)\propto\mathcal{N}(\mu_{ki},P_{ki}^{-1})$ are of
Gaussian form as well.

As the terms in the product of the incoming messages and the self
potential in the integral-product rule~(\ref{eq_contBP}) are all a
function of the same variable, $x_{i}$ (associated with the node
$i$), then, according to the multivariate extension of
Lemma~\ref{Lemma_Gaussian}, \mbox{$\phi_{i}(x_i) \prod_{k \in
\textrm{N}(i) \backslash j} m_{ki}(x_i)$} is proportional to a
certain Gaussian distribution, $\mathcal{N}(\mu_{i\backslash
j},P_{i\backslash j}^{-1})$. Applying the multivariate version of
the product precision expression in~(\ref{eq_productprec}), the
update rule for the inverse variance is given by (over-braces
denote the origin of each of the terms) \BE\label{eq_prec}
P_{i\backslash j} = \overbrace{P_{ii}}^{\phi_{i}(x_i)} +
\sum_{k\in \textrm{N}(i) \backslash j}
\overbrace{P_{ki}}^{m_{ki}(x_i)}, \EE where $P_{ii}\triangleq
A_{ii}$ is the inverse variance a-priori associated with node $i$,
via the precision of $\phi_{i}(x_{i})$, and $P_{ki}$ are the
inverse variances of the messages $m_{ki}(x_i)$. Similarly,
using~(\ref{eq_productmean}) for the multivariate case, we can
calculate the mean \BE\label{eq_mean}
 \mu_{i\backslash j} = P_{i\backslash j}^{-1}\Big(\overbrace{P_{ii}\mu_{ii}}^{\phi_{i}(x_i)} +
\sum_{{k} \in \textrm{N}(i) \backslash j}
\overbrace{P_{ki}\mu_{ki}}^{m_{ki}(x_i)}\Big), \EE where
$\mu_{ii}\triangleq b_{i}/A_{ii}$ is the mean of the self
potential and $\mu_{ki}$ are the means of the incoming messages.

Next, we calculate the remaining terms of the message
$m_{ij}(x_j)$, including the integration over $x_{i}$. After some
algebraic manipulation, using the Gaussian integral \mbox{$
\int_{-\infty}^{\infty}\exp{(-ax^{2}+bx)}dx=\sqrt{\pi/a}\exp{(b^{2}/4a)}$},
we find that the messages $m_{ij}(x_j)$ are proportional to a
normal distribution with precision and mean
\BEA\label{eq_prec_message}
P_{ij} = -A_{ij}^2P_{i \backslash j}^{-1},\\
\mu_{ij}\label{eq_mean_message} =
-P_{ij}^{-1}A_{ij}\mu_{i\backslash j}. \EEA These two scalars
represent the messages propagated in the GaBP-based algorithm.

Finally, computing the product rule~(\ref{eq_productrule}) is
similar to the calculation of the previous product and the
resulting mean~(\ref{eq_mean}) and precision~(\ref{eq_prec}), but
including all incoming messages. The marginals are inferred by
normalizing the result of this product. Thus, the marginals are
found to be Gaussian probability density functions
$\mathcal{N}(\mu_{i},P_{i}^{-1})$ with precision and mean \BEA
P_{i}
= \overbrace{P_{ii}}^{\phi_{i}(x_i)} + \sum_{k\in \textrm{N}(i)} \overbrace{P_{ki}}^{m_{ki}(x_i)},\\
\mu_{i} = P_{i\backslash
j}^{-1}\Big(\overbrace{P_{ii}\mu_{ii}}^{\phi_{i}(x_i)} + \sum_{{k}
\in \textrm{N}(i)}
\overbrace{P_{ki}\mu_{ki}}^{m_{ki}(x_i)}\Big)\label{eq_marginal_mean},
\EEA respectively.

For a dense data matrix, the number of messages passed on the
graph can be reduced from $\mathcal{O}(n^{2})$ (\ie, twice the
number of edges) down to $\mathcal{O}(n)$ messages per iteration
round by using a similar construction to
Bickson~\etal~\cite{BroadcastBP}: Instead of sending a unique
message composed of the pair of $\mu_{ij}$ and $P_{ij}$ from node
$i$ to node $j$, a node broadcasts aggregated sums to all its
neighbors, and consequently each node can retrieve locally
$P_{i\backslash j}$~(\ref{eq_prec}) and $\mu_{i\backslash
j}$~(\ref{eq_mean}) from the aggregated sums \BEA
\tilde{P}_{i}&=&P_{ii}+\sum_{{k}\in\textrm{N}(i)}
P_{ki},\\\tilde{\mu}_{i}&=&\tilde{P}_{i}^{-1}(P_{ii}\mu_{ii}+\sum_{k
\in \textrm{N}(i)} P_{ki}\mu_{ki}) \EEA by means of a subtraction
\BEA P_{i\backslash j}&=&\tilde{P}_{i}-P_{ji},\\\mu_{i\backslash
j}&=&\tilde{\mu}_{i}-P_{i \backslash j}^{-1}P_{ji}\mu_{ji}.\EEA
The following pseudo-code summarizes the GaBP solver
algorithm.\comment{The derivation of the GaBP solver algorithm, as
described in Section~\ref{sec_algorithm}, is concluded by simply
substituting the explicit expressions for $P_{i\backslash
j}$~(\ref{eq_prec}) into $P_{ij}$~(\ref{eq_prec_message}),
$\mu_{i\backslash j}$~(\ref{eq_mean}) and
$P_{ij}$~(\ref{eq_prec_message}) into
$\mu_{ij}$~(\ref{eq_mean_message}), and $P_{i\backslash
j}$~(\ref{eq_prec}) into $\mu_{i}$~(\ref{eq_marginal_mean}).}
%\end{proof}

\begin{alg}[GaBP
solver]\label{alg_GaBP_Broadcast}\end{alg}\vspace{-0.5cm}
\begin{center}\resizebox{0.5\textwidth}{!}{\begin{tabular}{|lll|}
  \hline&&\\
  \texttt{1.} & \emph{\texttt{Initialize:}} & $\checkmark$\quad\texttt{Set the neighborhood} $\textrm{N}(i)$ \texttt{to include}\\&&\quad\quad$\forall k\neq i \texttt{ such that } A_{ki}\neq0$.\\&& $\checkmark$\quad\texttt{Fix the scalars}\\&&$\quad\quad P_{ii}=A_{ii}$ \texttt{and} $\mu_{ii}=b_{i}/A_{ii}$, $\forall i$.\\
  && $\checkmark$\quad\texttt{Set the initial $i\rightarrow\textrm{N}(i)$ broadcast messages}\\&&\quad\quad $\tilde{P_{i}}=0$ \texttt{and} $\tilde{\mu}_{i}=0$.\\&&$\checkmark$\quad\texttt{Set the initial $k\rightarrow i,k\in\textrm{N}(i)$ internal scalars}\\&&\quad\quad $P_{ki}=0$ \texttt{and} $\mu_{ki}=0$.\\&& $\checkmark$\quad \texttt{Set a convergence threshold} $\epsilon$.\\% \hline
  {\texttt{2.}} & {\emph{\texttt{Iterate:}} }&  $\checkmark$\quad\texttt{Broadcast the aggregated sum messages}\\&&$\quad\quad \tilde{P}_{i}=P_{ii}+\sum_{{k}\in\textrm{N}(i)}
P_{ki}$,\\&&
$\quad\quad\tilde{\mu}_{i}=\tilde{P_{i}}^{-1}(P_{ii}\mu_{ii}+\sum_{k
\in \textrm{N}(i)} P_{ki}\mu_{ki})$, $\forall
i$\\&&\quad\quad\texttt{(under chosen scheduling)}.\\&&
  $\checkmark$\quad\texttt{Compute the} $i\rightarrow j,i\in\textrm{N}(j)$ \texttt{internal scalars} \\&& $\quad\quad P_{ij} = -A_{ij}^{2}/(\tilde{P}_{i}-P_{ji})$,\\
  &&$\quad\quad\mu_{ij}=(\tilde{P_{i}}\tilde{\mu_{i}}-P_{ji}\mu_{ji})/A_{ij}$.\\%\hline
  {\texttt{3.}} & {\emph{\texttt{Check:}}} & $\checkmark$\quad\texttt{If the internal scalars} $P_{ij}$ \texttt{and} $\mu_{ij}$ \texttt{did not}\\&&\quad\quad\texttt{converge (w.r.t. $\epsilon$),} \texttt{return to
    Step 2.}\\&&$\checkmark$\quad\texttt{Else, continue to Step 4.}\\%\hline
  {\texttt{4.}} & {\emph{\texttt{Infer:}}} & $\checkmark$\quad\texttt{Compute the marginal means}\\&&{\quad\quad$\mu_{i}=\big(P_{ii}\mu_{ii}+\sum_{k \in
\textrm{N}(i)}P_{ki}\mu_{ki}\big)/\big(P_{ii}+\sum_{{k}\in\textrm{N}(i)}
P_{ki}\big)=\tilde{\mu}_{i}$, $\forall i$.}\\&&$(\checkmark$\quad\texttt{Optionally compute the marginal precisions}\\&&\quad\quad$P_{i}=P_{ii}+\sum_{k\in\textrm{N}(i)}P_{ki}=\tilde{P}_{i}\quad)$\\
  %\hline
  {\texttt{5.}} & {\emph{\texttt{Solve:}}} & $\checkmark$\quad\texttt{Find the solution}\\&& {$\quad\quad x_{i}^{\ast}=\mu_{i}$, $\forall i$.}\\&&\\\hline
\end{tabular}}\end{center}

\subsection{Max-Product Rule}
A well-known alternative to the sum-product BP algorithm is the
max-product (a.k.a. min-sum) algorithm~\cite{BibDB:BookJordan}. In
this variant of BP, a maximization operation is performed rather
than marginalization, \ie, variables are eliminated by taking
maxima instead of sums. For trellis trees (\eg, graphical
representation of convolutional codes or ISI channels), the
conventional sum-product BP algorithm boils down to performing the
BCJR algorithm, resulting in the most probable symbol, while its
max-product counterpart is equivalent to the Viterbi algorithm,
thus inferring the most probable sequence of
symbols~\cite{BibDB:FactorGraph}.

In order to derive the max-product version of the proposed GaBP
solver, the integral(sum)-product rule~(\ref{eq_contBP}) is
replaced by a new rule\BE\label{eq_contBP2}
    m_{ij}(x_j)\propto\argmax_{x_i} \psi_{ij}(x_i,x_j) \phi_{i}(x_i)
\prod_{k \in \textrm{N}(i)\setminus j} m_{ki}(x_i). \EE Computing
$m_{ij}(x_j)$ according to this max-product rule, one gets (the
exact derivation is omitted) \BE
m_{ij}(x_j)\propto\mathcal{N}(\mu_{ij}=-P_{ij}^{-1}A_{ij}\mu_{i\backslash
j},P_{ij}^{-1} = -A_{ij}^{-2}P_{i \backslash j}), \EE which is
identical to the messages derived for the sum-product
case~(\ref{eq_prec_message})-(\ref{eq_mean_message}). Thus
interestingly, as opposed to ordinary (discrete) BP, the following
property of the GaBP solver emerges.
\begin{corol}[Max-product]
The max-product~(\ref{eq_contBP2}) and
sum-product~(\ref{eq_contBP}) versions of the GaBP solver are
identical.
\end{corol}

\section{Convergence and Exactness}
\label{sec:converge-exact} In ordinary BP, convergence does not
guarantee exactness of the inferred probabilities, unless the
graph has no cycles.   Luckily, this is not the case for the GaBP
solver. Its underlying Gaussian nature yields a direct connection
between convergence and exact inference. Moreover, in contrast to
BP, the convergence of GaBP is not limited to acyclic or sparse
graphs and can occur even for dense (fully-connected) graphs,
adhering to certain rules that we now discuss.
We can use results from the literature on probabilistic inference
in graphical
models~\cite{BibDB:Weiss01Correctness,BibDB:jmw_walksum_nips,BibDB:mjw_walksum_jmlr06}
to determine the convergence and exactness properties of the GaBP
solver. The following two theorems establish sufficient conditions
under which GaBP is guaranteed to converge to the exact marginal
means.

\begin{thm}{\cite[Claim 4]{BibDB:Weiss01Correctness}}
If the matrix $\mA$ is strictly diagonally dominant (\ie,
$|A_{ii}|>\sum_{j\neq
  i}|A_{ij}|,\forall i$), then GaBP
converges and the marginal means converge to the true means.
\end{thm}

This sufficient condition was recently relaxed to include a wider
group of matrices.

\begin{thm}{\cite[Proposition 2]{BibDB:jmw_walksum_nips}}
If the spectral radius (\ie, the maximum of the absolute values of
the eigenvalues) $\rho$ of the matrix $|\mI_{n}-\mA|$ satisfies
$\rho(|\mI_{n}-\mA|)<1$, then GaBP converges and the marginal
means converge to the true means.
\end{thm}

There are many examples of linear systems that violate these
conditions for which the GaBP solver nevertheless converges to the
exact solution. In particular, if the graph corresponding to the
system is acyclic (\ie, a tree), GaBP yields the exact marginal
means (and even marginal variances), regardless of the value of
the spectral radius~\cite{BibDB:Weiss01Correctness}.
\comment{However, in contrast to conventional iterative methods
derived from linear algebra, understanding the conditions for
exact convergence and quantifying the convergence rate of the GaBP
solver remain intriguing open problems.}

\section{Relation to Classical Solution Methods}
It can be shown (see also Plarre and
Kumar~\cite{BibDB:PlarreKumar}) that the GaBP solver
(Algorithm~\ref{alg_GaBP_Broadcast}) for a system of linear
equations represented by a tree graph is identical to the renowned
direct method of Gaussian elimination (a.k.a. LU
factorization,~\cite{BibDB:BookMatrix}). The interesting relation
to classical iterative solution methods~\cite{BibDB:BookAxelsson}
is revealed via the following proposition.
\begin{prop}[Jacobi and GaBP solvers]\label{prop_GaBPJ} \mbox{The GaBP solver} (Algorithm~\ref{alg_GaBP_Broadcast})
\begin{enumerate}
  \item with inverse variance messages arbitrarily set to zero, \ie, $P_{ij}=0, i\in\textrm{N}(j),\forall{j}$;
  \item incorporating the message received from node $j$ when computing the message to be sent from node $i$ to node $j$, \ie, replacing $k\in\textrm{N}(i)\backslash j$ with
  $k\in\textrm{N}(i)$;
\end{enumerate}
is identical to the Jacobi iterative method.
\end{prop}
\begin{proof}
Arbitrarily setting the precisions to zero, we get in
correspondence to the above derivation, \BEA
P_{i\backslash j}&=&P_{ii}=A_{ii},\\
P_{ij}\mu_{ij}&=&-A_{ij}\mu_{i\backslash j},\\
\label{eq_marginal_J}\mu_{i}&=&A_{ii}^{-1}(b_{i}-\sum_{k\in\textrm{N}(i)}A_{ki}\mu_{k\backslash
i}). \EEA Note that the inverse relation between $P_{ij}$ and
$P_{i\backslash j}$~(\ref{eq_prec_message}) is no longer valid in
this case.
Now, we rewrite the mean $\mu_{i\backslash j}$~(\ref{eq_mean})
without excluding the information from node $j$, \BE
\mu_{i\backslash
j}=A_{ii}^{-1}(b_{i}-\sum_{k\in\textrm{N}(i)}A_{ki}\mu_{k\backslash
i}). \EE Note that $\mu_{i\backslash j}=\mu_{i}$, hence the
inferred marginal mean $\mu_{i}$~(\ref{eq_marginal_J}) can be
rewritten as \BE\label{eq_J} \mu_{i}=A_{ii}^{-1}(b_{i}-\sum_{k\neq
i}A_{ki}\mu_{k}), \EE where the expression for all neighbors of
node $i$ is replaced by the redundant, yet identical, expression
$k\neq i$. This fixed-point iteration~(\ref{eq_J}) is identical to
the element-wise expression of the Jacobi
method\cite{BibDB:BookAxelsson}, concluding the proof.
\end{proof}

Now, the Gauss-Seidel (GS) method can be viewed as a `serial
scheduling' version of the Jacobi method; thus, based on
Proposition~\ref{prop_GaBPJ}, it can be derived also as an
instance of the serial (message-passing) GaBP solver. Next, since
successive over-relaxation (SOR) is nothing but a GS method
averaged over two consecutive iterations, SOR can be
obtained as a serial GaBP solver with `damping'
operation~\cite{Damping}.

\section{Application Example: Linear Detection}\label{sec_linear}
We examine the implementation of a decorrelator linear detector in
a CDMA system with spreading codes based upon Gold sequences of
length $N=7$. Two system setups are simulated,  corresponding to
$n=3$ and $n=4$ users.\comment{, resulting in the
cross-correlation matrices \BE
\mR_{n=3} = \frac{1}{7}\left(%
\begin{array}{rrr}
  7 & -1 & 3 \\
  -1 & 7 & -5 \\
  3 & -5 & 7 \\
\end{array} \right) \EE and
\BE \mR_{n=4} = \frac{1}{7}\left(%
\begin{array}{rrrr}
  7 & -1 & 3 & 3\\
  -1 & 7 & 3 & -1\\
  3 & 3 & 7 & -1\\
  3 & -1 & -1 & 7\\
\end{array} \right), \EE respectively.\footnote{These particular correlation settings were taken from the simulation setup of Yener
\etal~\cite{BibDB:YenerEtAl}.}}
The decorrelator detector, a member of the family of linear
detectors, solves a system of linear equations,
\mbox{$\mA\vx=\vb$}, where the matrix $\mA$ is equal to the
$n\times n$ correlation matrix $\mR$, and the observation vector
$\vb$ is identical to the $n$-length CDMA channel output vector
$\vy$.  Thus, the vector of decorrelator decisions is determined
by taking the signum (for binary signaling) of the vector
$\mA^{-1}\vb = \mR^{-1}\vy$. Note that $\mR_{n=3}$ and $\mR_{n=4}$
in this case are not strictly diagonally dominant, but their
spectral radii are less than unity, since
\mbox{$\rho(|\mI_{3}-\mR_{n=3}|)=0.9008<1$} and
\mbox{$\rho(|\mI_{4}-\mR_{n=4}|)=0.8747<1$}, respectively. In all
of the experiments, we assumed the (noisy) output sample was the
all-ones vector.

\begin{table}[t!] \centerline{
\resizebox{0.4\textwidth}{!}{\begin{tabular}{|c|r|r|}
  \hline
  \textbf{Algorithm}
  & Iterations $t$ ($\mR_{n=3}$) & Iterations $t$ ($\mR_{n=4}$) \\\hline\hline & & \\
  Jacobi & 111\comment{without dividing R_{3} with 7: 136} & 24\comment{50} \\\hline & & \\
  GS & 26 & 26\comment{32}\\\hline & & \\
  \textbf{Parallel GaBP} & \textbf{23} & \textbf{24}\\\hline & & \\
  Optimal SOR & 17 & 14\comment{20} \\\hline & & \\
  \textbf{Serial GaBP} & \textbf{16} & \textbf{13}\\
  \hline & &\\Jacobi+Steffensen & 59\comment{51} & $-$ \\\hline & & \\
  \textbf{Parallel GaBP+Steffensen} & \textbf{13} & \textbf{13}\\\hline & & \\
  \textbf{Serial GaBP+Steffensen} & \textbf{9} & \textbf{7} \\
  \hline
\end{tabular}
}}\vspace{0.1cm}\caption{Convergence rate.}\label{tab_1}
\end{table}

Table~\ref{tab_1} compares the proposed GaBP solver with standard
iterative solution methods~\cite{BibDB:BookAxelsson}, previously
employed for CDMA multiuser detection (MUD). Specifically, MUD
algorithms based on the algorithms of Jacobi, GS and (optimally
weighted) SOR were
investigated~\cite{BibDB:YenerEtAl,grant99iterative,BibDB:TanRasmussen}.
Table~\ref{tab_1} lists the convergence rates for the two Gold
code-based CDMA settings. Convergence is identified and declared
when the differences in all the iterated values are less than
$10^{-6}$. We see that, in comparison with the previously proposed
detectors based upon the Jacobi and GS algorithms, the serial
(asynchronous) message-passing GaBP detector converges more
rapidly for both $n=3$ and $n=4$ and achieves the best overall
convergence rate, surpassing even the optimal SOR-based detector.
Further speed-up of the GaBP solver can be achieved by adopting
known acceleration techniques from linear algebra.
\comment{Consider a sequence $\{x_{n}\}$ (\eg, obtained by using
GaBP iterations) linearly converging to the limit $\hat{x}$, and
$x_n \ne \hat{x}$ for $n \ge 0$. According to Aitken's method, if
there exists a real number $a$ such that $|a|<1 $ and
\mbox{$\lim_{n \rightarrow \infty}(x_n-\hat{x})/(x_{n-1} -
\hat{x}) = a$}, then the sequence $\{ y_n\}$ defined by
\[ y_n = x_n - \frac{(x_{n+1} -x_n)^2}{x_{n+2} - 2x_{n+1} + x_n} \]
converges to $\hat{x}$ faster than $\{ x_n \}$ in the sense that
\mbox{$\lim_{n \rightarrow \infty} |(\hat{x} - y_n)/(\hat{x} -
x_n)| = 0$}. Aitken's method can be viewed as a generalization of
over-relaxation, since one uses values from three, rather than
two, consecutive iteration rounds. This method can be easily
implemented in GaBP as every node computes values based only on
its own history.

Steffensen's iterations incorporate Aitken's method. Starting with
$x_{n}$, two iterations are run to get $x_{n+1}$ and $x_{n+2}$.
Next, Aitken's method is used to compute $y_{n}$, this value
replaces the original $x_{n}$, and GaBP is executed again to get a
new value of $x_{n+1}$. This process is repeated iteratively until
convergence.} Table~\ref{tab_1} demonstrates the speed-up of the
GaBP solver obtained by using such an acceleration method, termed
Steffensen's iterations~\cite{BibDB:BookHenrici}, in comparison
with the accelerated Jacobi algorithm (diverged for the 4 users
setup). We remark that this is the first time such an acceleration
method is examined within the framework of message-passing
algorithms and that the region of convergence of the accelerated
GaBP solver remains unchanged.
\begin{figure}[t!]
\begin{center}
\includegraphics[width=0.4\textwidth]{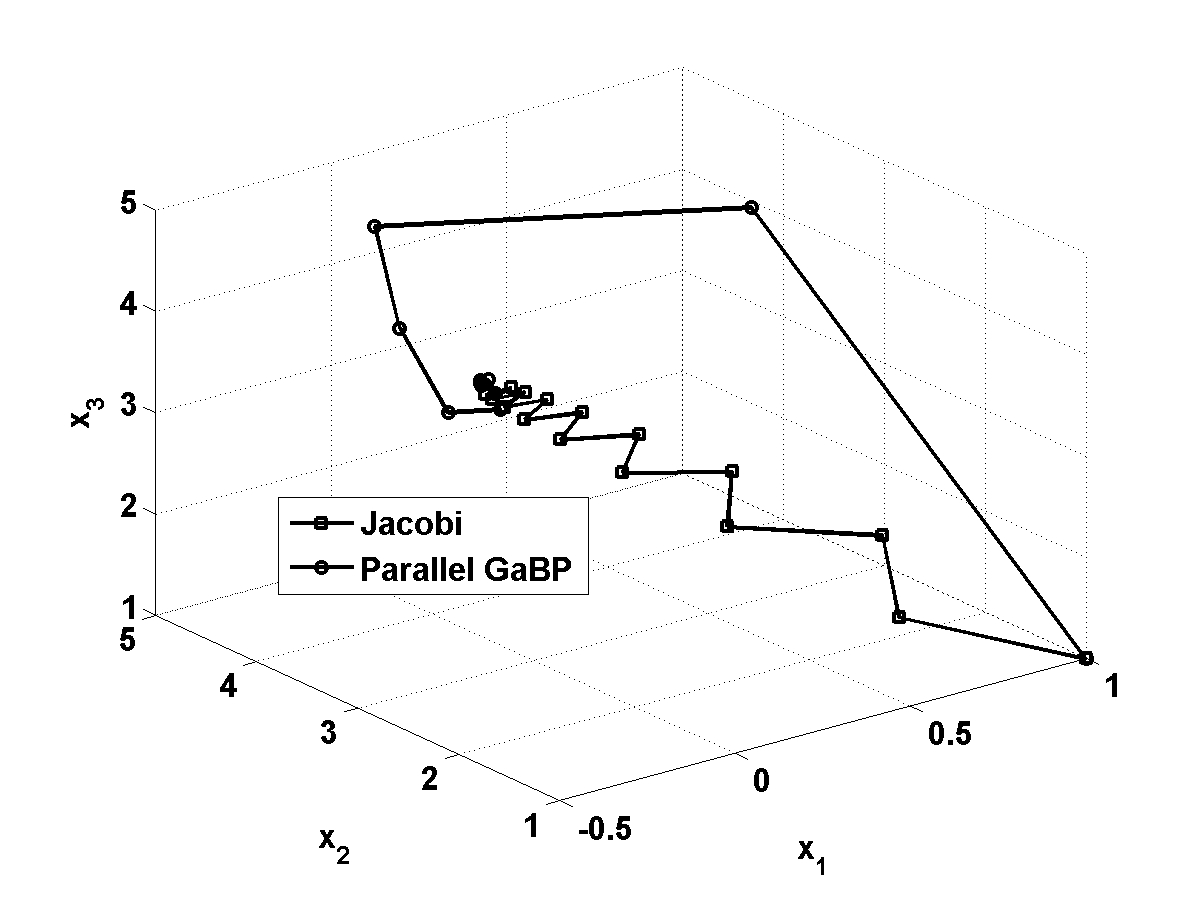}\end{center}
\vspace{-0.5cm}\caption{Convergence visualization.}\label{fig_1}
\end{figure}

The convergence contours for the Jacobi and parallel (synchronous)
GaBP solvers for the case of 3 users are plotted in the space of
$\{x_{1},x_{2},x_{3}\}$ in Fig.~\ref{fig_1}. As expected, the
Jacobi algorithm converges in zigzags directly towards the fixed
point\comment{(this behavior is well-explained in Bertsekas and
Tsitsiklis~\cite{BibDB:BookBertsekasTsitsiklis})}. It is
interesting to note that the GaBP solver's convergence is in a
spiral shape, hinting that despite the overall convergence
improvement, performance improvement is not guaranteed in
successive iteration rounds. Further results and elaborate
discussion on the application of GaBP specifically to linear MUD
may be found in recent contributions~\cite{Allerton,ISIT2}.
\vspace{-0.2cm}
\bibliographystyle{IEEEtran}   % (uses file "plain.bst")
\bibliography{IEEEabrv,isit_Solver_final}       % expects file "myrefs.bib"

\begin{thebibliography}{10}
\providecommand{\url}[1]{#1}
\csname url@rmstyle\endcsname
\providecommand{\newblock}{\relax}
\providecommand{\bibinfo}[2]{#2}
\providecommand\BIBentrySTDinterwordspacing{\spaceskip=0pt\relax}
\providecommand\BIBentryALTinterwordstretchfactor{4}
\providecommand\BIBentryALTinterwordspacing{\spaceskip=\fontdimen2\font plus
\BIBentryALTinterwordstretchfactor\fontdimen3\font minus
  \fontdimen4\font\relax}
\providecommand\BIBforeignlanguage[2]{{%
\expandafter\ifx\csname l@#1\endcsname\relax
\typeout{** WARNING: IEEEtran.bst: No hyphenation pattern has been}%
\typeout{** loaded for the language `#1'. Using the pattern for}%
\typeout{** the default language instead.}%
\else
\language=\csname l@#1\endcsname
\fi
#2}}

\bibitem{BibDB:BookMatrix}
G.~H. Golub and C.~F.~V. Loan, \emph{Matrix Computation}, 3rd~ed.\hskip 1em
  plus 0.5em minus 0.4em\relax Baltimore, MD: The Johns Hopkins University
  Press, 1996.

\bibitem{BibDB:BookAxelsson}
O.~Axelsson, \emph{Iterative Solution Methods}.\hskip 1em plus 0.5em minus
  0.4em\relax Cambridge, UK: Cambridge University Press, 1994.

\bibitem{BibDB:BookSaad}
Y.~Saad, \emph{Iterative Methods for Sparse Linear Systems}.\hskip 1em plus
  0.5em minus 0.4em\relax PWS Publishing company, 1996.

\bibitem{BibDB:BookPearl}
J.~Pearl, \emph{Probabilistic Reasoning in Intelligent Systems: Networks of
  Plausible Inference}.\hskip 1em plus 0.5em minus 0.4em\relax San Francisco:
  Morgan Kaufmann, 1988.

\bibitem{BibDB:BookJordan}
M.~I. Jordan, Ed., \emph{Learning in Graphical Models}.\hskip 1em plus 0.5em
  minus 0.4em\relax Cambridge, MA: The MIT Press, 1999.

\bibitem{BibDB:BookMCT}
T.~Richardson and R.~Urbanke, \emph{Modern Coding Theory}.\hskip 1em plus 0.5em
  minus 0.4em\relax Cambridge University Press, 2007.

\bibitem{BibDB:FactorGraph}
F.~Kschischang, B.~Frey, and H.~A. Loeliger, ``Factor graphs and the
  sum-product algorithm,'' \emph{{IEEE} Trans. Inform. Theory}, vol.~47, pp.
  498--519, Feb. 2001.

\bibitem{BibDB:Weiss01Correctness}
Y.~Weiss and W.~T. Freeman, ``Correctness of belief propagation in {Gaussian}
  graphical models of arbitrary topology,'' \emph{Neural Computation}, vol.~13,
  no.~10, pp. 2173--2200, 2001.

\bibitem{BroadcastBP}
\BIBentryALTinterwordspacing
D.~Bickson, D.~Dolev, and Y.~Weiss, ``Modified belief propagation for energy
  saving in wireless and sensor networks,'' in \emph{Leibniz Center TR-2005-85,
  School of Computer Science and Engineering, The Hebrew University}, 2005.
  [Online]. Available: \url{http://leibniz.cs.huji.ac.il/tr/842.pdf}
\BIBentrySTDinterwordspacing

\bibitem{BibDB:jmw_walksum_nips}
J.~K. Johnson, D.~M. Malioutov, and A.~S. Willsky, ``Walk-sum interpretation
  and analysis of {Gaussian} belief propagation,'' in \emph{Advances in Neural
  Information Processing Systems 18}, Y.~Weiss, B.~Sch\"{o}lkopf, and J.~Platt,
  Eds.\hskip 1em plus 0.5em minus 0.4em\relax Cambridge, MA: MIT Press, 2006,
  pp. 579--586.

\bibitem{BibDB:mjw_walksum_jmlr06}
D.~M. Malioutov, J.~K. Johnson, and A.~S. Willsky, ``Walk-sums and belief
  propagation in {Gaussian} graphical models,'' \emph{Journal of Machine
  Learning Research}, vol.~7, Oct. 2006.

\bibitem{BibDB:PlarreKumar}
K.~Plarre and P.~Kumar, ``Extended message passing algorithm for inference in
  loopy {Gaussian} graphical models,'' \emph{Ad Hoc Networks}, 2004.

\bibitem{Damping}
K.~M. Murphy, Y.~Weiss, and M.~I. Jordan, ``Loopy belief propagation for
  approximate inference: {An} empirical study,'' in \emph{Proc. of {UAI}},
  1999.

\bibitem{BibDB:YenerEtAl}
A.~Yener, R.~D. Yates, and S.~Ulukus, ``{CDMA} multiuser detection: {A}
  nonlinear programming approach,'' \emph{{IEEE} Trans. Commun.}, vol.~50,
  no.~6, pp. 1016--1024, June 2002.

\bibitem{grant99iterative}
A.~Grant and C.~Schlegel, ``Iterative implementations for linear multiuser
  detectors,'' \emph{{IEEE} Trans. Commun.}, vol.~49, no.~10, pp. 1824--1834,
  Oct. 2001.

\bibitem{BibDB:TanRasmussen}
P.~H. Tan and L.~K. Rasmussen, ``Linear interference cancellation in {CDMA}
  based on iterative techniques for linear equation systems,'' \emph{{IEEE}
  Trans. Commun.}, vol.~48, no.~12, pp. 2099--2108, Dec. 2000.

\bibitem{BibDB:BookHenrici}
P.~Henrici, \emph{Elements of Numerical Analysis}.\hskip 1em plus 0.5em minus
  0.4em\relax New York: John Wiley and Sons, 1964.

\bibitem{Allerton}
D.~Bickson, O.~Shental, P.~H. Siegel, J.~K. Wolf, and D.~Dolev, ``Linear
  detection via belief propagation,'' in \emph{Proc. 45th Allerton Conf. on
  Communications, Control and Computing}, Monticello, IL, USA, Sept. 2007.

\bibitem{ISIT2}
------, ``Gaussian belief propagation based multiuser detection,'' in
  \emph{IEEE Int. Symp. on Inform. Theory (ISIT)}, Toronto, Canada, July 2008.

\end{thebibliography}

\end{document}